\documentclass[11pt,letterpaper]{article}
\usepackage[utf8]{inputenc}
\usepackage[english]{babel}
\usepackage{latexsym,amsbsy,amssymb,amsmath,amsfonts}
\usepackage{amsthm}
\usepackage{epsfig}
\usepackage{euscript}
\usepackage{fullpage}
\usepackage{multirow,multicol,booktabs}
\usepackage[linesnumbered,ruled,noend]{algorithm2e}
\usepackage{nicefrac}
\usepackage{tikz}
\usetikzlibrary{calc}
\usetikzlibrary{decorations.pathreplacing}
\usepackage[bf]{caption}
\usepackage{graphicx}
\usepackage{enumitem}
\usepackage{hyperref}
\hypersetup{
 pdfauthor={Christian Glaßer, Christian Reitwießner, Maximilian Witek},
 pdftitle={Applications of Discrepancy Theory in Multiobjective Approximation}
  unicode,
  breaklinks,
  colorlinks=false,
  pdfborder={0 0 0}
}

\date{}

\newcommand{\Sum}{\sum\limits}

\newtheorem{dummytheorem}{Dummy-Theorem}[section]

\newtheorem{lemma}[dummytheorem]{Lemma}
\newtheorem{theorem}[dummytheorem]{Theorem}

\newtheorem{corollary}[dummytheorem]{Corollary}

\newtheorem{claim}[dummytheorem]{Claim}

\newcommand{\uint}{\mathbb{N}}
\newcommand{\rational}{\mathbb{Q}}

\newcommand{\oli}[1]{\overline{#1}}

\newcommand{\eps}{\varepsilon}
\newcommand{\card}[1]{\##1}
\newcommand{\length}[1]{\mathrm{length}({#1})}

\newcommand{\pn}[1]{\textnormal{#1}}

\newcommand{\algosettings}{
\SetAlgoNoLine
\SetKwSty{texttt}
\SetFuncSty{texttt}
\SetArgSty{texttt}
\SetDataSty{texttt}
\SetAlFnt{\texttt}
\SetNlSty{texttt}{\normalsize}{}
\NoCaptionOfAlgo
\SetKwInOut{Input}{\textrm{Input}}
\SetKwInOut{Output}{\textrm{Output}}
}

\newcommand{\kSTSP}{\pn{$k$-MaxSTSP}}

\newcommand{\kATSP}{\pn{$k$-MaxATSP}}

\newcommand{\kTSPApprox}{\pn{Alg-$k$-MaxTSP}}

\newcommand{\kMaxSATbf}{\pn{\bf $\boldsymbol{k}$-MaxSAT}}
\newcommand{\kMaxSAT}{\pn{$k$-MaxSAT}}
\newcommand{\kWSATApprox}{\pn{Alg-$k$-MaxSAT}}

\newcommand{\FixedDCC}{{$\mathrm{MaxDCC_F}$}}
\newcommand{\FixedUCC}{{$\mathrm{MaxUCC_F}$}}
\newcommand{\bfFixedDCC}{{$\boldsymbol{\mathrm{MaxDCC_F}}$}}
\newcommand{\bfFixedUCC}{{$\boldsymbol{\mathrm{MaxUCC_F}}$}}

\newcommand{\norm}[1]{||{#1}||}

\begin{document}
\selectlanguage{english}

\title{Applications of Discrepancy Theory in Multiobjective Approximation}

\author{
Christian Glaßer \hspace*{6mm} 
Christian Reitwießner \hspace*{6mm} 
Maximilian Witek\\[5.7mm] 
{Julius-Maximilians-Universität Würzburg, Germany} \\
{\tt \{glasser,reitwiessner,witek\}@informatik.uni-wuerzburg.de}}

\maketitle

\begin{abstract}
We apply a multi-color extension of the Beck-Fiala theorem to show
that the multiobjective maximum traveling salesman problem is
randomized $\nicefrac{1}{2}$-approximable on directed graphs and
randomized $\nicefrac{2}{3}$-approximable on undirected graphs.
Using the same technique we show that the multiobjective maximum
satisfiablilty problem is $\nicefrac{1}{2}$-approximable.
\end{abstract}

%%%%%%%%%%%%%%%%%%%%%%%%%%%%%%%%%%%%%%%%%%%%%%%%%%%%%%%%%%%%%%%%%%%%%%%%%%%%%%%%
%%%%%%%%%%%%%%%%%%%%%%%%%%%%%%%%%%%%%%%%%%%%%%%%%%%%%%%%%%%%%%%%%%%%%%%%%%%%%%%%
%%%%%%%%%%%%%%%%%%%%%%%%%%%%%%%%%%%%%%%%%%%%%%%%%%%%%%%%%%%%%%%%%%%%%%%%%%%%%%%%

\section{Introduction}

We study multiobjective variants of
the traveling salesman problem and the satisfiability problem.
\begin{itemize}
    \item
    The $k$-objective maximum traveling salesman problem:
    Given is a directed/undirected complete graph with edge weights from $\uint^k$.
    Find a Hamiltonian cycle of maximum weight.
    \item
    The $k$-objective maximum weighted satisfiability problem:
    Given is a Boolean formula in conjunctive normal form and
    for each clause a non-negative weight in $\uint^k$.
    Find a truth assignment that maximizes the sum of the weights
    of all satisfied clauses.
\end{itemize}
In general we cannot expect to find a single solution
that is optimal with respect to all objectives.
Instead we are interested in the {\em Pareto set}
which consists of all optimal solutions in the sense that
there is no solution that is at least as good in all objectives and
better in at least one objective.
Typically, the Pareto set has exponential size,
and this particularly holds for the traveling salesman
and the satisfiability problems considered here.
We are hence interested in computing an approximation of the Pareto set.

A popular strategy for approximating single-objective traveling salesman and
single-objective satisfiability is to compute two or more alternatives out of which
one chooses the best one:
\begin{itemize}
\item For each cycle in a maximum cycle cover of a graph, 
remove the edge with \emph{the lowest weight},
and connect the remaining paths to a Hamiltonian cycle.
\item For some formula, 
take an arbitrary truth assignment and its complementary truth assignment,
and return the one with \emph{the highest weight} of satisfied clauses.
\end{itemize}
However, in the presence of multiple objectives,
these alternatives can be incomparable and hence we need an argument
that allows to \emph{appropriately combine} incomparable alternatives.

While previous work focused on problem-specific properties
to construct solutions of good quality,
we show that the Beck-Fiala theorem~\cite{BF81} from discrepancy theory 
and its multi-color extension due to Doerr and Srivastav~\cite{DS03}
provide a general and simple way to combine alternatives appropriately.
Its application leads to simplified and improved approximation algorithms for the
$k$-objective maximum traveling salesman problem on directed and undirected
graphs
and the $k$-objective maximum weighted satisfiability problem.

%%%%%%%%%%%%%%%%%%%%%%%%%%%%%%%%%%%%%%%%%%%%%%%%%%%%%%%%%%%%%%%%%%%%%%%%%%%%%%%%
%%%%%%%%%%%%%%%%%%%%%%%%%%%%%%%%%%%%%%%%%%%%%%%%%%%%%%%%%%%%%%%%%%%%%%%%%%%%%%%%
%%%%%%%%%%%%%%%%%%%%%%%%%%%%%%%%%%%%%%%%%%%%%%%%%%%%%%%%%%%%%%%%%%%%%%%%%%%%%%%%

\section{Preliminaries}

\paragraph{Multiobjective Optimization}
Let $k \geq 1$ and consider some $k$-objective maximization
problem $\mathcal{O}$ that consists of a set of instances $\mathcal{I}$, a set of solutions
$S(x)$ for each instance $x \in \mathcal{I}$, and a function $w$ assigning a
$k$-dimensional weight $w(x, s) \in \uint^k$ to each solution $s \in S(x)$
depending also on the instance $x \in \mathcal{I}$. If the instance $x$ is clear
from the context, we also write $w(s) = w(x, s)$.  The components of $w$ are
written as $w_i$.
For weights $a = (a_1 , \dots,
a_k)$, $b = (b_1 , \dots, b_k ) \in \uint^k$ we write $a \ge b$ if $a_i \ge b_i$
for all $i$.

Let $x \in \mathcal{I}$.
The Pareto set of $x$, the set of all optimal solutions, is the set $\{s \in S(x)
\mid \neg \exists s' \in S(x)\,\, (w(x, s') \ge w(x, s) 
\text{ and } w(x, s') \neq w(x, s))\}$. For solutions $s,
s' \in S(x)$ and $\alpha < 1$ we say $s$ is $\alpha$-approximated by $s'$ if
$w_i(s') \ge \alpha \cdot w_i (s)$ for all $i$. We call a set of solutions
$\alpha$-approximate Pareto set of $x$ if every solution $s \in S(x)$ (or equivalently, every
solution from the Pareto set) is $\alpha$-approximated by some $s'$ contained in
the set.

We say that some algorithm is an $\alpha$-approximation algorithm for
$\mathcal{O}$ if
it runs in polynomial time and returns an $\alpha$-approximate Pareto set of $x$
for all inputs $x \in \mathcal{I}$. We call it randomized if it is
allowed to fail with probability at most $\nicefrac{1}{2}$.
An algorithm is an FPTAS (fully polynomial-time approximation scheme) for
$\mathcal{O}$,
if on input $x$ and $0<\eps<1$ it computes a
$(1-\eps)$-approximate Pareto set of $x$ in time polynomial in
$\nicefrac{1}{\eps} + \length{x}$.
If the algorithm is randomized it is called FPRAS
(fully polynomial-time randomized approximation scheme).
If for each $\eps$, the runtime of such an algorithm is polynomial in
$\length{x}$, we call it PTAS (polynomial-time approximation scheme)
or PRAS (polynomial-time randomized approximation scheme).

\paragraph{Graph Prerequisites}
An \emph{$\uint^k$-labeled directed (undirected) graph} 
is a tuple $G=(V,E,w)$, 
where $V$ is some finite set of vertices, 
$E \subseteq V \times V$ ($E \subseteq {V \choose 2}$)
is a set of directed (undirected) edges, 
and $w \colon E \to \uint^k$ is a $k$-dimensional weight function.
If $E= (V \times V) \setminus \{(i,i) \mid i \in V\}$
($E = {V \choose 2}$) then $G$ is called \emph{complete}.
We denote the $i$-th component of $w$ by $w_i$ 
and extend $w$ to sets of edges by taking the sum over the weights of all edges in the set. 
A \emph{cycle (of length $m \ge 1$)} in $G$ is an alternating sequence of vertices and edges $v_0,e_1,v_1, \dots e_{m},v_m$, 
where $v_i \in V$, $e_j \in E$,
$e_j=(v_{j-1},v_j)$ ($e_j=\{v_{j-1},v_j\}$) for all $0 \leq i \leq m$ and $1 \leq j \leq m$,
neither the sequence of vertices $v_0,v_1,\dots,v_{m-1}$ nor the sequence of
edges $e_1,\dots,e_m$ contains any repetition, and $v_m=v_0$.
A cycle in $G$ is called \emph{Hamiltonian}
if it visits every vertex in $G$.
A set of cycles in $G$ is called \emph{cycle cover}
if for every vertex $v \in V$ it contains exactly one cycle that visits $v$.
For simplicity we interpret cycles and cycle covers as sets of edges
and can thus (using the above mentioned extension of $w$ to sets of edges) 
write $w(C)$ for the (multidimensional) weight of a cycle cover $C$ of $G$.

\paragraph{Approximating Cycle Covers} We will consider approximation algorithms for the multiobjective traveling salesman problem
using a multiobjective version of the maximum cycle cover problem.
For directed input graphs we have the following problem definition.

\begin{quote}
\begin{tabbing}
\textbf{$\boldsymbol{k}$-Objective Maximum Directed Edge-Fixed $\boldsymbol{c}$-Cycle Cover}
\textbf{($\boldsymbol{k}$-$\boldsymbol{c}$-\bfFixedDCC)}\\
Instance: \= $\uint^k$-labeled complete directed graph $(V,E,w)$ and $F \subseteq E$\\
Solution: \> Cycle cover $C \subseteq E$ with at least $c$ edges per cycle and $F \subseteq C$\\
Weight: \> $w(C)$
\end{tabbing}
\end{quote}

For undirected input graphs we analogously define the
$k$-objective maximum undirected edge-fixed $c$-cycle cover problem
(\textbf{$\boldsymbol{k}$-$\boldsymbol{c}$-\bfFixedUCC}, for short).
Let 
 \textbf{$\boldsymbol{k}$-$\boldsymbol{c}$-UCC}
(\textbf{$\boldsymbol{k}$-$\boldsymbol{c}$-DCC})
denote the problems we obtain from 
$k$-$c$-\FixedDCC\ (\textbf{$\boldsymbol{k}$-$\boldsymbol{c}$-\bfFixedUCC})
if we require $F=\emptyset$.
Using this notation we obtain the usual cycle cover problems
\textbf{$\boldsymbol{k}$-MaxDCC} as $k$-$0$-DCC and
\textbf{$\boldsymbol{k}$-MaxUCC} as $k$-$0$-UCC.

Manthey and Ram \cite{MR09} show by a reduction to matching 
that there is an FPRAS for $k$-objective \emph{minimum} cycle cover problems.
The same technique can be used to
show that there are FPRAS for $k$-MaxDCC and $k$-MaxUCC \cite{man09}.
We show that there are FPRAS for $k$-$2$-\FixedDCC\ and $k$-$3$-\FixedUCC\ by
a reduction to $k$-MaxDCC and $k$-MaxUCC.

\begin{theorem}\label{theoremkcc} For every $k \ge 1$,
$k$-$2$-\FixedDCC\ and $k$-$3$-\FixedUCC\ admit an FPRAS.
\end{theorem}

\begin{proof}
For every $l \ge 1$,
let $l$-MaxDCC-Approx ($l$-MaxUCC-Approx) 
denote the FPRAS for $l$-MaxDCC ($l$-MaxUCC).
We begin with the directed case.

Let $k \ge 1$.
On input of the $\mathbb{N}^k$-labeled complete directed graph $G=(V,E,w)$ 
and $F \subseteq E$,
let $G'=(V,E,w')$,
where $w' \colon E \to \mathbb{N}^{k+1}$
such that for all $e \in E$,
\begin{align*}
w'_i(e) &= w_i(e) &\text{for $1 \le i \le k$}\\
w'_{k+1}(e) &=
\begin{cases}
1 &\text{ if } e \in F \\
0 &\text{ otherwise.}
\end{cases}
\end{align*}
For $\eps>0$, apply $(k+1)$-MaxDCC-Approx to $G'$ with approximation ratio
$\eps' = \min\{\eps, \nicefrac{1}{(r+1)}\}$, where $r := \card{F}$
and return the obtained set of cycle covers that contain all edges from $F$.

Let $C$ be some (arbitrary) cycle cover with $F \subseteq C$.  If no
such cycle cover exists, we are done. Otherwise,
we have $w'_{k+1}(C) = r$, and with probability at least $\nicefrac{1}{2}$
the FPRAS must have returned some cycle cover $C'$ that $\eps'$-approximates $C$.
By $\eps' \leq \nicefrac{1}{(r+1)}$
we have $w'_{k+1}(C') 
\geq (1-\eps') \cdot w'_{k+1}(C) 
\geq (1-\nicefrac{1}{(r+1)}) \cdot r
= r - \nicefrac{r}{(r+1)} 
> r - 1$
and hence $F \subseteq C'$.
Moreover, by $\eps' \leq \eps$ 
we have $w_i(C') = w'_i(C')
\geq (1-\eps') \cdot w'_i(C)
\geq (1-\eps) \cdot w'_i(C)
= (1-\eps) \cdot w_i(C)$ 
for all $1 \leq i \leq k$.
Since an arbitrary cycle in a complete directed graph has length at least two,
the assertion is proved.

The proof for the undirected case is very similar,
as we call $(k+1)$-MaxUCC-Approx instead. 
Since in a complete undirected graph every cycle has length at least three, 
the assertion follows.
\end{proof}

\paragraph{Boolean Formulas}
We consider formulas over a finite set of propositional variables $V$,
where a \emph{literal} is a propositional variable $v\in V$ or its negation $\overline{v}$,
a \emph{clause} is a finite, nonempty set of literals,
and a \emph{formula in conjunctive normal form} (\textbf{CNF}, for short)
is a finite set of clauses.
A \emph{truth assignment} is a mapping $I \colon V \to \{0,1\}$.
For some $v \in V$,
we say that $I$ 
\emph{satisfies the literal $v$} if $I(v)=1$,
and $I$ \emph{satisfies the literal $\overline{v}$} if $I(v)=0$.
We further say that $I$ \emph{satisfies the clause $C$} and write $I(C)=1$ 
if there is some literal $l \in C$ that is satisfied by $I$.

%%%%%%%%%%%%%%%%%%%%%%%%%%%%%%%%%%%%%%%%%%%%%%%%%%%%%%%%%%%%%%%%%%%%%%%%%%%%%%%%
%%%%%%%%%%%%%%%%%%%%%%%%%%%%%%%%%%%%%%%%%%%%%%%%%%%%%%%%%%%%%%%%%%%%%%%%%%%%%%%%
%%%%%%%%%%%%%%%%%%%%%%%%%%%%%%%%%%%%%%%%%%%%%%%%%%%%%%%%%%%%%%%%%%%%%%%%%%%%%%%%

\section{Multi-Color Discrepancy}

Suppose we have a list of items with (single-objective) weights and want to find a subset of
these items with about half of the total weight. The exact
version of this problem is of course the NP-complete problem \textsc{partition} \cite{GJ79},
and hence it is unlikely that an exact solution can be found in polynomial time.
If we allow a deviation in the order of the largest weight,
this problem can be solved in polynomial time, though. 
Surprisingly, this is
still true if the weights are not single numbers but vectors of numbers, which
follows from a classical result in discrepancy theory known as the Beck-Fiala
theorem~\cite{BF81}.  It is important to note that
the allowed deviation is independent of the number of vectors
since this enables us to use this result in
multiobjective approximation for balancing out multiple objectives at the same
time with an error that does not depend on the input size.

In the Beck-Fiala theorem and the task discussed above, we have to decide for
each item to either include it or not. In some situations in multiobjective
optimization, though, a more general problem needs to be solved:
There is a constant number of weight vectors for each item,
out of which we have to choose exactly one.
Doerr and Srivastav~\cite{DS03} showed that 
the Beck-Fiala theorem generalizes to this so-called multi-color setting
with almost the same deviation.
Their proof implicitly shows that this choice 
can be computed in polynomial time. For completeness we restate the proof
and argue for polynomial-time computability.

For a vector $x \in \rational^m$ let $\norm{x}_\infty = \max_i{|x_i|}$,
and for a matrix $A\in \rational^{m \times n}$ let $\norm{A}_1 = \max_j\sum_i |a_{ij}|$.
For $c \ge 2$, $n \ge 1$ let $\overline{M_{c,n}} = \{x \in
(\rational\cap[0,1])^{cn} \mid \sum_{k=0}^{c-1} x_{cb-k} = 1 $ for all $b \in
\{1,\dots,n\}\}$ and $M_{c,n} =
\overline{M_{c,n}} \cap \{0,1\}^{cn}$.

\begin{theorem}[Doerr, Srivastav \cite{DS03}]\label{theoremmulticolordisc} There is a polynomial-time algorithm that on
input of some $A \in \rational^{m \times cn}$, $m,n \in\uint$, $c \ge 2$ and $p
\in \overline{M_{c,n}}$ finds a
coloring $\chi \in M_{c,n}$ such that $\norm{A(p-\chi)}_\infty
\le 2\norm{A}_1$.
\end{theorem}
\begin{proof}
Let $\Delta := \norm{A}_1$.
We start with $\chi=\chi^{(0)}=p\in\overline{M_{c,n}}$ and
will successively change it to a vector in $M_{c,n}$. We will first describe the
algorithm and then argue about its runtime.

Let $J := J(\chi) := \{j \in \{1,\dots,cn\} \mid \chi_j \notin
\{0,1\}\}$ and call the columns from $J$ floating.
Let $I := I(\chi) := \{i \in \{1,\dots,m\} \mid \sum_{j \in
J(\chi)} |a_{ij}| > 2 \Delta\}$. We will ensure that during the rounding process the following
conditions are fulfilled (this is clear from the start, because $\chi^{(0)} = p$):

\begin{align*}
(A(p-\chi))_{|I} = 0 \qquad \text{(C1)} && \chi \in \overline{M_{c,n}} \qquad \text{(C2)}
\end{align*}

Let us assume that the rounding process is at step $t$ where the current coloring
is $\chi=\chi^{(t)}$ and the conditions (C1) and (C2) hold. If there is no
floating column, i.e., $J=\emptyset$,
then $\chi \in M_{c,n}$ and thus 
$\chi$ has the desired form.

Otherwise, assume that there are still floating columns. Let $B = \{b \in
\{1,\dots,n\} \mid \exists k \in \{0,\dots,c-1\}\colon cb-k \in J\}$ be the
$c$-blocks that contain floating columns. Since $\chi \in \overline{M_{c,n}}$,
a $c$-block of $\chi$ contains either none or at least two floating columns, thus
$\card{B} \le \frac{1}{2}\card{J}$.

Since
\begin{align*}
\card{J}\cdot \Delta
= \sum_{j\in J} \Delta
\ge \sum_{j\in J}\sum_{i=1}^m |a_{ij}|
\ge \sum_{j\in J}\sum_{i \in I} |a_{ij}|
= \sum_{i \in I}\sum_{j \in J} |a_{ij}|
> \sum_{i \in I} 2 \Delta =
 \card{I} \cdot 2 \Delta
\end{align*}
it holds that $\card{I} < \frac{1}{2}\card{J}$. Consider the inhomogeneous
system of linear equations
\begin{align*}
&(A (p - \chi))_{|I} = 0\\
&\sum_{k=0}^{c-1}\chi_{cb-k}=1 && \text{for } b \in B
\end{align*}
where each $\chi_j$ is considered as a variable if $j \in J$ and as a constant
if $j \notin J$.
This system consists of at most $\card{I} + \card{B} < \frac{1}{2} \card{J} +
\frac{1}{2}\card{J} = \card{J}$ equations and $\card{J}$ variables and
hence is under-determined.
Note that the system has the solution $\chi_{|J}$
because $\chi$ fulfills the conditions (C1) and (C2). Since
it is under-determined, it also has a second solution $x\in
\rational^J$. We extend $x$
to $x_E \in \rational^{cn}$ by
\[
(x_E)_j = \begin{cases} x_j &\text{if } j \in J \\ \chi_j
&\text{otherwise.}\end{cases}
\]
Consider the line $\{(1-\lambda)\,\chi + \lambda \, x_E \mid \lambda \in
\rational\}$. Each point on this line (or rather its restriction to the components in
$J$) fulfills the system of equations and thus condition (C1).
By condition (C2) and the definition of $J$ it holds that $0 < \chi_j < 1$ for
all $j \in J$ and thus there is some $\lambda \in \rational$ such
that
$\chi^{(t+1)}:=(1-\lambda)\,\chi + \lambda \, x_E \in \overline{M_{c,n}}$ and at least one
component becomes $0$ or $1$, i.e., $J(\chi^{(t+1)}) \subsetneq J(\chi^{(t)})$.
Note that
$\chi^{(t+1)}$ fulfills (C1) and (C2)
even for the larger sets $J(\chi^{(t)})$ and $I(\chi^{(t)})$.
Continue the rounding process with $\chi := \chi^{(t+1)}$.

Since at least one column is removed from $J$ in each iteration, the rounding process will
eventually stop. Let $\chi$ be the final value of the coloring.
We show $\norm{A(p-\chi)}_\infty \le 2\Delta$. Let $1
\le i \le m$. Since at the end, $J=\emptyset$ we also have $I = \emptyset$.
Let $\chi^{(t)}$ be the first coloring such that $i \notin I$. Since
$\chi^{(t)}$ fulfills (C1) also for $I(\chi^{(t-1)})$ (or $I(\chi^{(0)})$ if
$t=0$) we
have $(A(p - \chi^{(t)}))_i = 0$. Furthermore it holds that
$\chi^{(t)}_j = \chi_j$ for all $j \notin
J(\chi^{(t)})$ and $|\chi^{(t)}_j - \chi_j|< 1$ for all $j \in J(\chi^{(t)})$.
Finally note that
$\sum_{j \in J(\chi^{(t)})} |a_{ij}| \le 2\Delta$ since $i\notin I(\chi^{(t)})$.
Combining these facts, we obtain
\[
|(A(p - \chi))_i| =
|(A(p - \chi^{(t)}))_i + (A(\chi^{(t)} - \chi))_i| =
| 0 + \sum_{j\in J(\chi^{(t)})} a_{ij}(\chi^{(t)}_j - \chi_j)| \le
2\Delta.
\]

We now analyse the runtime. Note that we have at most $cn$ iterations, which is
polynomial in the input length. In each iteration, we have to solve an
inhomogeneous system of linear equations
and we have to find a certain $\lambda \in \rational$. The system, whose size is polynomial in
the input length, can be solved in polynomial
time (see for instance \cite[Theorem 1.4.8]{GLS88}). By adding an equation of the form
$\chi_j=2$ for some suitable $j\in J$, we can find a solution different to $\chi$. The value for $\lambda$ can be obtained
in polynomial time by successively trying to fix each floating column to $0$ or $1$, solving
for $\lambda$ and checking if the resulting vector is still in
$\overline{M_{c,n}}$.
\end{proof}

\begin{corollary}\label{coromulticolordisc}
There is a polynomial-time algorithm that on input of a set of vectors $v^{j,r} \in
\rational^m$ for $1 \le j \le n$, $1 \le r \le c$ computes a
coloring $\chi\colon\{1,\dots,n\}\to\{1,\dots,c\}$ such that for each $1 \le i
\le m$ it holds that
\[
\left|
\frac{1}{c}\sum_{j=1}^n\sum_{r=1}^c v^{j,r}_i
-
\sum_{j=1}^n v^{j,\chi(j)}_i
\right| \le 2m\max_{j,r} |v^{j,r}_i|.
\]
\end{corollary}
\begin{proof}
The result is obvious for $c = 1$. For $c \ge 2$, we use
Theorem~\ref{theoremmulticolordisc}. Because the error bound is different for
each row, we need to scale the rows of the vectors. Let $\delta_i = \max_{j,r}
|v^{j,r}_i|$ for $1 \le i \le m$.
Let $A=(a_{i,j'}) \in \rational^{m\times cn}$ where
$a_{i,(c(j-1)+r)} = \frac{1}{\delta_i} v^{j,r}_i$ (if $\delta_i=0$, set it to
$0$) and $p \in \rational^{cn}$ such that
$p_i=\frac{1}{c}$ for all $1 \le i \le cn$.
We obtain a coloring $\chi \in
\{0,1\}^{cn}$ such that for each $1 \le j \le n$ there is exactly one $1 \le r
\le c$ such that $\chi_{c(j-1)+r} = 1$ and it holds that
$\norm{A(p-\chi)}_\infty \le 2\norm{A}_1$. Note that because of the scaling, the
largest entry in $A$ is $1$ and thus we have $\norm{A}_1 \le m$. Define
$\chi'\colon\{1,\dots,n\}\to\{1,\dots,c\}$ by $\chi'(j)=r \iff \chi_{c(j-1)+r} =
1$. We obtain for each $1 \le i \le m$:
\begin{align*}
2m\delta_i \ge 2 \norm{\delta_i \, A}_1 \ge |(\delta_i\, A(p-\chi))_i| =
\left|\sum_{j'=1}^{cn}\delta_i a_{ij'}(p_{j'} -
\chi_{j'})\right|
= \left|\sum_{j=1}^n\sum_{r=1}^c\frac{1}{c}v^{j,r}_i
- 
\sum_{j=1}^nv^{j,\chi'(j)}_i\right|
\end{align*}
\end{proof}

%%%%%%%%%%%%%%%%%%%%%%%%%%%%%%%%%%%%%%%%%%%%%%%%%%%%%%%%%%%%%%%%%%%%%%%%%%%%%%%%
%%%%%%%%%%%%%%%%%%%%%%%%%%%%%%%%%%%%%%%%%%%%%%%%%%%%%%%%%%%%%%%%%%%%%%%%%%%%%%%%
%%%%%%%%%%%%%%%%%%%%%%%%%%%%%%%%%%%%%%%%%%%%%%%%%%%%%%%%%%%%%%%%%%%%%%%%%%%%%%%%

\section{Approximation of Multiobjective Maximum Traveling Salesman}

\paragraph{Definition}
Given some complete $\uint^k$-labeled graph as input,
our goal is to find a Hamiltonian cycle of maximum weight.
For directed graphs this problem is called 
$k$-objective maximum asymmetric traveling salesman ($k$-MaxATSP),
while for undirected graphs it is called
$k$-objective maximum symmetric traveling salesman ($k$-MaxSTSP).
Below we give the formal definition of $k$-MaxATSP,
the problem $k$-MaxSTSP is defined analogously.

\begin{quote}
\begin{tabbing}
\textbf{$\boldsymbol{k}$-Objective Maximum Asymmetric Traveling Salesman ($\boldsymbol{k}$-MaxATSP)}\\ 
Instance: \= $\uint^k$-labeled directed complete graph $(V,E,w)$\\
Solution: \> Hamiltonian cycle $C$\\
Weight: \> $w(C)$
\end{tabbing}
\end{quote}

\paragraph{Previous Work}
In 1979, Fisher, Nemhauser and Wolsey \cite{FNW79} gave a $\nicefrac{1}{2}$-approximation
algorithm for the single-objective maximum asymmetric traveling salesman ($1$-MaxATSP)
by removing the lightest edge from each cycle of a maximum cycle cover
and connecting the remaining paths to a Hamiltonian cycle.
Since undirected cycles always contain at least three edges, 
this also showed that the single-objective maximum symmetric traveling salesman ($1$-MaxSTSP)
is $\nicefrac{2}{3}$-approximable.
Since then, many improvements were achieved, 
and currently, the best known approximation ratios of 
$\nicefrac{2}{3}$ for $1$-MaxATSP and $\nicefrac{7}{9}$ for $1$-MaxSTSP
are due to Kaplan et al.\ \cite{KLS+05}
and Paluch, Mucha and Madry~\cite{PMM09}.

Most single-objective approximation algorithms 
do not directly translate to the case of multiple objectives,
and hence we need more sophisticated algorithms.
For $k$-MaxATSP and $k$-MaxSTSP, where $k \geq 2$,
the currently best known approximation algorithms are due to Manthey,
who showed a randomized $(\nicefrac{1}{2}-\eps)$-approximation of $k$-MaxATSP
and a randomized $(\nicefrac{2}{3}-\eps)$-approximation of $k$-MaxSTSP~\cite{man09}.
Recently, Manthey also showed 
a deterministic $(\nicefrac{1}{2k}-\eps)$-approximation of $k$-MaxSTSP
and a deterministic $(\nicefrac{1}{(4k-2)}-\eps)$-approximation of $k$-MaxATSP~\cite{man11}.

\paragraph{Our Results}
We show that $k$-MaxATSP is randomized $\nicefrac{1}{2}$-approximable
and $k$-MaxSTSP is randomized $\nicefrac{2}{3}$-approximable
using the following idea.
We choose a suitable
number $l$ depending only on $k$ and try all sets of at most $l$ edges $F$ using brute force. 
For each such $F$ we apply the FPRAS for $k$-$2$-\FixedDCC\ ($k$-$3$-\FixedUCC),
which exists by Theorem~\ref{theoremkcc},
fixing the edges in $F$.
For all cycle covers thus obtained, 
we select two (three) edges from each cycle and
compute a 2-coloring (3-coloring) of the cycles with low discrepancy with regard to the weight
vectors of the selected edges.
Using this coloring, we remove exactly one edge from each cycle and connect the
remaining simple paths to a single cycle in an arbitrary way. Since the coloring has low
discrepancy, we only remove about one half (one third) of the weight in each objective.
The introduced error is absorbed by choosing suitable heavy edges $F$ at the
beginning. The described procedure generally works for arbitrary $c$-cycle covers.

\begin{algorithm}[htbp]
\algosettings
\caption{\textbf{Algorithm}: \texttt{\kTSPApprox($V,E,w$)} with parameter $c
\ge 2$}
\Input{$\uint^{k}$-labeled directed/undirected complete graph $G=(V,E,w)$}
\Output{set of Hamiltonian cycles of $G$}
\BlankLine
\ForEach{\textnormal{$F_H, F_L \subseteq E \text{ with } \card{F_H}\le 3\, c\,k^2$, $\card{F_L} \leq
c\,\card{F_H}$\label{algo1:loop:heavysets}}}
{
   let $\delta \in \mathbb{N}^k$ with $\delta_i = \max\{n \in \mathbb{N} \mid 
      \text{there are $3\,c\,k$ edges $e \in F_H$ with $w_i(e) \geq n$} \}$\label{algo1:line:delta}\;
   \ForEach{$e \in E \setminus F_H$\label{algo1:loop:ClearWrongEdges}}
   {
      \lIf{$w(e) \not\leq \delta$}
      {
         modify $w$ such that $w(e)=0$ 
         for the current iteration of line~\ref{algo1:loop:heavysets}\;
      }
   }
   compute $(1-\nicefrac{1}{\card{V}})$-approximation $\mathcal{S}$ of
   $k$-$c$-\FixedDCC\ / $k$-$c$-\FixedUCC\ on $(G,F_H \cup F_L)$
      \label{algo1:line:cyclecovers}\;
   \ForEach{\textnormal{cycle cover $S \in \mathcal{S}$\label{algo1:loop:LoopAllCycleCovers}}}
   {
      let $C_1,\dots,C_r$ denote the cycles in $S$\label{algo1:line:CyclesInS}\;
      \If{\textnormal{for each $i \in \{1,\dots,r\}$, $C_i \setminus F_H$ contains
      a path of length $c$\label{algo1:lightedgesincycle}}}{
          \lForEach{$i\in\{1,\dots,r\}$\label{algo1:loop:markpaths}}
          {
             choose path $e_{i,1},\dots,e_{i,c} \in C_i \setminus F_H$ arbitrarily\;
          }
          compute some coloring $\chi \colon \{1,\dots,r\} \to \{1,\dots,c\}$ such that
          \begin{eqnarray*}
             \Sum_{i=1}^r w(e_{i,\chi(i)})
             & \leq & 2k \cdot \delta + \frac{1}{c}\Sum_{i=1}^r\sum_{j=1}^c w(e_{i,j}) 
          \end{eqnarray*}
             and remove the edges $\{e_{i,\chi(i)} \mid i=1,\dots,r\}$ from $S$\label{algo1:line:BeckFialaApplication}\;
          output the remaining edges, arbitrarily connected to a Hamiltonian cycle\;
      }
   }
}
\end{algorithm}

\begin{lemma}\label{lemmaalgorithm}
Let $c \ge 2$ and $k \ge 1$.
If there exists an FPRAS for $k$-$c$-\FixedDCC\ ($k$-$c$-\FixedUCC, resp.),
then the algorithm \kTSPApprox\ computes a randomized $(1-\nicefrac{1}{c})$-approximation
for \kATSP\ (\kSTSP, resp.).
\end{lemma}

\begin{proof}
Let $k \geq 1$, $c \geq 2$, and $G=(V,E,w)$ be some $\mathbb{N}^k$-labeled (directed or
undirected) input graph
with $m=\card{V}$ sufficiently large.

We will first argue that the algorithm terminates 
in time polynomial in the length of $G$.
Since there are only polynomially many subsets $F_H,F_L \subseteq E$
with cardinality bounded by a constant,
the loop in line~\ref{algo1:loop:heavysets}
is executed polynomially often.
In each iteration the FPRAS on $G=(V,E,w)$ and $F_H \cup F_L \subseteq E$ 
terminates in time polynomial in the length of $G$ and $F_H\cup F_E$,
which means that the set $\mathcal{S}$ contains only polynomially many cycle covers.
Hence, for each iteration of the loop in line~\ref{algo1:loop:heavysets},
the loop in line~\ref{algo1:loop:LoopAllCycleCovers}
is also executed at most polynomially many times,
and overall we have polynomially many nested iterations.
In each nested iteration where each cycle of the cycle cover contains a path as required,
we compute a coloring of $\{1,\dots,r\}$ with low discrepancy.
By Corollary~\ref{coromulticolordisc} this can be done in polynomial time.
Observe that all further steps require at most polynomial time, 
and hence the algorithm terminates after polynomially many steps.

Next we argue that the algorithm will succeed with probability at least $\nicefrac{1}{2}$.
Observe that the only randomized parts of the algorithm
are the calls to the randomized cycle cover approximation algorithm
in line~\ref{algo1:line:cyclecovers}. Using amplification we can
assume that the probability that all the calls to this algorithm succeed is at least
$\nicefrac{1}{2}$.

It remains to show that if the algorithm \kTSPApprox\ succeeds, 
it outputs some $(1-\nicefrac{1}{c})$-approximate set of Hamiltonian
cycles.
Hence, for the remainder of the proof, let us assume that the algorithm
and hence all calls to the internal FPRAS succeed.
Furthermore, let $R \subseteq E$ be some Hamiltonian cycle of $G$.
We will argue that there is some iteration where the algorithm 
outputs an $(1-\nicefrac{1}{c})$-approximation of $R$.

For each $1 \leq i \leq k$, let $F_{H,i} \subseteq R$ 
be some set of $3\,c\,k$ heaviest edges of $R$ in the $i$-th component,
breaking ties arbitrarily.
Let $F_H = \bigcup_{i=1}^k F_{H,i}$. We define $F_L \subseteq R$ such that
$F_L \cap F_H = \emptyset$ and each
edge in $F_H$ is part of a path in
$F_L\cup F_H$ that contains $c$ edges from $F_L$. This is always possible as long as
$R$ is large enough. We now have 
$\card{F_H} \leq 3\,c\,k^2$ and $\card{F_L} \le c\,\card{F_H}$.
Hence in line~\ref{algo1:loop:heavysets} there will be some iteration that
chooses $F_H$ and $F_L$.
We fix this iteration for the remainder of the proof.

Let $\delta \in \uint^k$ as defined in line~\ref{algo1:line:delta}
and observe that $\delta_i = \min\{w_i(e) \mid e \in F_{H,i}\}$ for all $i$,
which means that for all edges $e \in R \setminus F_H$ we have $w(e) \leq \delta$.
Hence the loop in line~\ref{algo1:loop:ClearWrongEdges}
sets the weights of all edges $e \in E \setminus R$ 
that do not fulfill $w(e) \leq \delta$ to zero,
and these are the only weights that are modified.
In particular, this does not affect edges in $R$,
hence $w(R)$ remains unchanged. Note that since we do not increase the weight of
any edge and do not change the weight of the edges in $R$, it suffices to show
that the algorithm computes an approximation with respect to the changed
weights.

Next we obtain a $(1-\nicefrac{1}{\card{V}})$-approximate set 
$\mathcal{S}$ of $c$-cycle covers of $G$ that contain $F_H \cup F_L$.
Since $R$ is a $c$-cycle cover of $G$ with $F_H \cup F_L \subseteq R$,
there must be some $c$-cycle cover $S \in \mathcal{S}$ with $F_H \cup F_L \subseteq S$
that $(1-\nicefrac{1}{\card{V}})$-approximates $R$.
Hence in line~\ref{algo1:loop:LoopAllCycleCovers} 
there will be some iteration that chooses this $S$.
Again we fix this iteration for the remainder of the proof.

As in line~\ref{algo1:line:CyclesInS}, let $C_1,\dots,C_r$ denote the cycles in $S$.
Note that each cycle contains at least $c$ edges. Since each edge in $F_H$ is
part of a path in $F_H \cup F_L$ with at least $c$ edges from $F_L$, we even
know that each cycle contains at least $c$ edges not from $F_H$ and thus the
condition in line~\ref{algo1:lightedgesincycle} is fulfilled. Let these edges
$e_{i,j}$ be defined as in the algorithm. Note that since $e_{i,j} \notin F_H$
we have $w(e_{i,j}) \le \delta$ for all $i,j$, because the weight function was
changed accordingly.

In line~\ref{algo1:line:BeckFialaApplication} we compute some $\chi \colon
\{1,\dots,r\}\to\{1,\dots,c\}$
such that
\begin{eqnarray*}
 \Sum_{i=1}^r w(e_{i,\chi(i)}) & \leq &
 2\,k \cdot \delta + \frac{1}{c} \Sum_{i=1}^r\sum_{j=1}^c w(e_{i,j}) \\
  & \leq &
 2\,k\cdot\delta + \frac{1}{c} \cdot w(S \setminus F_H).
\end{eqnarray*}
Recall that by 
Corollary~\ref{coromulticolordisc}
such a coloring exists and can be computed in polynomial time.
Removing the chosen edges breaks the cycles into simple paths,
which can be arbitrarily connected to a Hamiltonian cycle $R'$.
For the following estimation note that
$\delta \le \frac{w(F_H)}{3\,c\,k}$ and
$w(F_H) \ge \frac{3\,c\,k}{m}w(R)$ and recall that $m = \card{V} = \card{R}$.
\allowdisplaybreaks
\begin{eqnarray*}
w(R')
& \geq & w(S) - \Sum_{i=1}^r w(e_{i,\chi(i)})\\
& \geq & w(S) - 2\,k\cdot\delta - \frac{1}{c} \cdot w(S \setminus F_H)\\
& = & \left(1-\frac{1}{c}\right)w(S) + \frac{1}{c}w(F_H) - 2\,k\cdot\delta\\
& \ge & \left(1-\frac{1}{c}\right)w(S) + \frac{1}{3\,c}w(F_H)\\
& \ge & \left(1-\frac{1}{c}\right)\left(1-\frac{1}{m}\right)w(R) +
\frac{k}{m} w(R)\\
& = & \left(1-\frac{1}{c}\right)w(R) + \left(-\left(1-\frac{1}{c}\right) +
k\right)\frac{1}{m} w(R)\\
& \ge & \left(1-\frac{1}{c}\right)w(R)
\end{eqnarray*}
This proves the assertion.
\end{proof}

It is known that 1-$c$-DCC is APX-hard for all $c \ge 3$ \cite{BM05} and that
1-$c$-UCC is APX-hard for $c \ge 5$ \cite{Man08}.
This means that, unless P $=$ NP, there is no PTAS for these
problems \cite{crka99} (and especially not for the variants with fixed edges). Furthermore, the existence of an FPRAS or
PRAS for these problems implies NP $=$ RP and thus a collapse of the
polynomial-time hierarchy, which is seen as follows.

If an APX-hard problem has a PRAS, then all problems in APX have a PRAS
and hence MAX-3SAT has one.
There exists an $\eps > 0$ and a polynomial-time computable $f$
mapping CNF formulas to 3-CNF formulas such that
if $x \in \mbox{SAT}$, then $f(x) \in \mbox{3SAT}$; and
if $x \notin \mbox{SAT}$, then there is no assignment satisfying
more than a fraction of $1-\eps$ of $f(x)$'s clauses \cite[Theorem 10.1]{al97}.
The PRAS for MAX-3SAT allows us to compute probabilistically a
$(1-\nicefrac{\eps}{2})$-approximation for $f(x)$ which in turn
tells us whether or not $x \in $ SAT. Since this procedure has no
false negatives we get RP $=$ NP, which implies
a collapse of the polynomial-time hierarchy \cite{lau83,sip83}.

So it seems unlikely that there is a PRAS for
1-$c$-DCC where $c \ge 3$ and 1-$c$-UCC where $c \ge 5$.
However, this does not necessarily mean that
the above algorithm is useless for parameters $c \ge 3$ in the directed and $c
\ge 5$ in the undirected case: The algorithm could still benefit from a
constant-factor approximation for  $k$-$c$-\FixedUCC\ or $k$-$c$-\FixedDCC.
A simple change in the
estimation shows that if the cycle cover algorithm has an approximation ratio
of $\alpha$, the above algorithm provides an approximation with ratio
$\alpha (1-\nicefrac{1}{c})$.

\begin{theorem} Let $k \ge 1$.
\begin{enumerate}
\item $\kATSP$ is randomized $\nicefrac{1}{2}$-approximable.
\item $\kSTSP$ is randomized $\nicefrac{2}{3}$-approximable.
\end{enumerate}
\end{theorem}
\begin{proof}
We combine Theorem~\ref{theoremkcc} and 
Lemma~\ref{lemmaalgorithm}.
\end{proof}

%%%%%%%%%%%%%%%%%%%%%%%%%%%%%%%%%%%%%%%%%%%%%%%%%%%%%%%%%%%%%%%%%%%%%%%%%%%%%%%%
%%%%%%%%%%%%%%%%%%%%%%%%%%%%%%%%%%%%%%%%%%%%%%%%%%%%%%%%%%%%%%%%%%%%%%%%%%%%%%%%
%%%%%%%%%%%%%%%%%%%%%%%%%%%%%%%%%%%%%%%%%%%%%%%%%%%%%%%%%%%%%%%%%%%%%%%%%%%%%%%%

\section{Approximation of Multiobjective Maximum Satisfiability}

\paragraph{Definition}
Given a formula in CNF and a function that maps each clause
to a $k$-objective weight, our goal is to find truth assignments
that maximize the sum of the weights of all satisfied clauses.
The formal definition is as follows.
\begin{quote}
    \begin{tabbing}
        {\bf $\boldsymbol{k}$-Objective Maximum Weighted Satisfiability
        (\kMaxSATbf)}\\
        Instance: \= Formula $H$ in CNF over a set of variables $V$,
        weight function $w \colon H \to \uint^k$\\
        Solution: \> Truth assignment $I \colon V \to \{0,1\}$\\
        Weight: \> Sum of the weights of all clauses satisfied by $I$, i.e.,
         $w(I) = \Sum_{\substack{C \in H\\I(C) = 1}} w(C)$
    \end{tabbing}
\end{quote}

\paragraph{Previous Work}
The first approximation algorithm for maximum satisfiability is due to Johnson \cite{Joh74},
whose greedy algorithm showed that the single-objective $1$-MaxSAT problem 
is $\nicefrac{1}{2}$-approximable.
Further improvements on the approximation ratio followed, 
and the currently best known approximation ratio of $0.7846$ for $1$-MaxSAT
is due to Asano and Williamson \cite{AW02}.

Only little is known about $k$-MaxSAT for $k\geq 2$.
Santana et.\ al.\ \cite{SBLL09} apply genetic algorithms to a version of the problem that is
equivalent to $k$-MaxSAT with polynomially bounded weights.
To our knowledge, the approximability of $k$-MaxSAT for $k \geq 2$
has not been investigated so far.

\paragraph{Our Results}
We show that $k$-MaxSAT is $\nicefrac{1}{2}$-approximable mainly by transferring the
idea that for any truth assignment, the assignment itself or its complementary
assignment satisfies at least one half of all clauses
to multidimensional objective functions.
We choose some suitable parameter $l \in \uint$ 
depending only on the number of objectives.
For a given formula in CNF we try all possible partial truth assignments for each
set of at most $l$ variables using brute force and extend each partial assignment to a full asignment in the following way:
For each remaining variable $v$ we compute two vectors roughly representing the
weight gained by the two possible assignments for $v$.
We then compute a $2$-coloring of those weight vectors with low discrepancy
which completes the partial assignment to a truth assignment
whose weight is at least one half of the total weight of the remaining
satisfiable clauses minus some error.
This error can be compensated by choosing $l$ large enough such that the partial
assignment already contributes a large enough weight.
This results in a $\nicefrac{1}{2}$-approximation for \kMaxSAT.

\paragraph{Notations}
For a set of clauses ${H}$ and a variable $v$ let ${H}[v=1] =
\{C \in {H} \mid v \in C\}$ be the set of clauses 
that are satisfied if this variable
is assigned one, and analogously ${H}[v=0] =
\{C \in {H} \mid \oli{v} \in C\}$ be the set of clauses
that are satisfied if this variable is assigned zero. This notation is extended to
sets of variables $V$ by ${H}[V=i]=\bigcup_{v \in V}{H}[v=i]$ for
$i=0,1$.

\begin{algorithm}[H]
    \algosettings
    \caption{\textbf{Algorithm}: \texttt{\kWSATApprox($H,w$)}}
    \Input{Formula $H$ in CNF over the variables $V=\{v_1,\dots,v_m\}$,
     $k$-objective weight function $w \colon H \to \uint^{k}$}
    \Output{Set of truth assignments $I \colon V \to \{0,1\}$}
    \BlankLine

    \ForEach{\textnormal{disjoint $V^0,V^1 \subseteq V$ with $\card{(V^0 \cup
    V^1)} \le 4k^2$}}{
        $G := H \setminus (H[V^0=0] \cup H[V^1=1])$\;
        $\hat{V}^{(1-i)} := \{v \in V \setminus (V^0 \cup V^1) \mid 
        4k \cdot w(G[v=i])\not\leq w(H \setminus G)\}$, $i=0,1$\;
        \If{$\hat{V}^0 \cap \hat{V}^1=\emptyset$}{
            $V' := V \setminus (V^0 \cup V^1 \cup \hat{V}^0 \cup \hat{V}^1)$, $L' := V' \cup \{\oli{v} \mid v \in V'\}$\;
            $G' := (G[V'=0] \cup G[V'=1]) \setminus (G[\hat{V}^0=0] \cup G[\hat{V}^1=1])$\;
            for $v_j \in V'$ let
            $x^{j,i}=\sum \{ \frac{w(C)}{\card{(C \cap L')}} \mid C \in G'[v_j=i]\}$
            for $i =0,1$\;
            compute some coloring $\chi\colon V' \to \{0,1\}$ such that\label{satalgolinecomputecoloring}
            \[
                \sum_{v_j \in V'} x^{j,\chi(j)} \ge \frac{1}{2}\sum_{v_j \in
                V'}(x^{j,0}+x^{j,1}) - 2k \delta
            \]
            where $\delta_r = \max\{x^{j,i}_r \mid v_j \in V',i\in\{0,1\}\}$\;
            let $I(v) := i$ for $v \in V^i\cup\hat{V}^i \cup \chi^{-1}(\{i\})$,
            $i=0,1$\;
            output $I$
        }
    }
\end{algorithm}

\begin{theorem}
    \kMaxSAT\ is $\nicefrac{1}{2}$-approximable
    for any $k \geq 1$.
\end{theorem}
\begin{proof}
We show that this approximation is realized by \kWSATApprox.
First note that this algorithm runs in polynomial time since $k$ is constant and
the coloring in line~\ref{satalgolinecomputecoloring} can be computed in
polynomial time using Corollary~\ref{coromulticolordisc}.
For the correctness, let $(H,w)$ be the input where $H$ is a formula over the
variables $V=\{v_1,\dots,v_m\}$ and $w \colon H \to \uint^{k}$ is the
$k$-objective weight function.
Let $I_o\colon V \rightarrow \{0,1\}$ be an optimal truth assignment.
We show that there is a loop iteration of \kWSATApprox($H,w$)
that outputs a truth assignment $I$ such that $w(I) \ge w(I_o)/2$.
To this end, we first show that there is an iteration of the loop that uses suitable
sets $V^0$ and $V^1$.

\begin{claim}\label{claimv0}
There are sets $V^i \subseteq I_o^{-1}(\{i\})$, $i=0,1$
with $\card{(V^0 \cup V^1)} \le 4k^2$ such that for
$G = H \setminus (H[V^0=0] \cup H[V^1=1])$ and any
$v \in V \setminus (V^0 \cup V^1)$ it holds that
\begin{equation} \label{eqn_762963}
w(G[v=I_o(v)]) \leq \frac{1}{4k} w(H \setminus G).
\end{equation}
\end{claim}
\begin{proof}
The assertion obviously holds for $\card{V} \le 4k^2$.
Otherwise, we define variables $
\{u_{kt+r} \in V \mid r=1,2,\dots,k$ and $ t=0,1,\dots,4k-1\}$,
sequentially in a greedy fashion:

\texttt{%
$H_0 := H$\\
for $t := 0$ to $4k-1$:\\
\hspace*{1em} for $r := 1$ to $k$:\\
\hspace*{2.5em} $j := kt+r$\\
\hspace*{2.5em} choose $v \in V \setminus \{u_1,\dots,u_{j-1}\}$ such that
$w_r(H_{j-1}[v=I_o(v)])$ is maximal\\
\hspace*{2.5em} $u_j := v$, $H_j := H_{j-1}\setminus H_{j-1}[v=I_o(v)]$, $\alpha_{j} := w(H_{j-1}[v=I_o(v)])$.
}

Let $V^i=I_o^{-1}(\{i\}) \cap \{u_1,\dots,u_{4k^2}\}$ for $i=0,1$.  
We now show inequality~(\ref{eqn_762963}), so let
$v \in V \setminus (V^0 \cup V^1)$. Assume that there is some $r \in
\{1,\dots,k\}$ such that $w_r(G[v=I_o(v)]) > \frac{1}{4k} w_r(H \setminus G)$.
Because the union $\bigcup_{j=1}^{4k^2}
H_{j-1}[u_j=I_o(u_j)] = H \setminus G$ is disjoint, we get
\begin{align*}
w(H \setminus G) = \sum_{r'=1}^{k}\sum_{t=0}^{4k-1} \alpha_{kt+r'} \ge
\sum_{t=0}^{4k-1} \alpha_{kt+r}
\intertext{and thus}
w_r(G[v=I_o(v)]) > \sum_{t=0}^{4k-1} \frac{(\alpha_{kt+r})_r}{4k}.
\end{align*}
Hence, by the pigeonhole principle, there must be some $t \in \{0,1,\dots,4k-1\}$ such that
$w_r(G[v=I_o(v)]) > (\alpha_{kt+r})_r$. But since $G \subseteq H_{kt+r-1}$
and thus even $w_r(H_{kt+r-1}[v=I_o(v)]) \ge w_r(G[v=I_o(v)]) >
(\alpha_{kt+r})_r$, the variable $v$ should have been chosen in step $j=kt+r$,
which is a contradiction. This means that $w(G[v=I_o(v)]) \le
\frac{1}{4k}w(H\setminus G)$ holds for all $v \in V \setminus (V^0\cup V^1)$.
\end{proof}

We choose the iteration of the algorithm where $V^0$ and $V^1$ equal the sets
whose existence is guaranteed
by Claim~\ref{claimv0}. In the following, we use the variables as they are
defined in the algorithm.
Observe that by the claim it holds that $I_o(v)=i$ for all $v \in \hat{V}^{i}$
for $i=0,1$.
Note that
\begin{align*}
\sum_{v_j \in V'} x^{j,0} + x^{j,1} =
        \sum_{v_j\in V'}\sum_{i\in\{0,1\}}\sum_{C \in G'[v_j=i]} \frac{w(C)}{\card{(C \cap L')}}
= \sum_{C \in G'} \card{(C \cap L')} \frac{w(C)}{\card{(C \cap L')}}
 = w(G').
\end{align*}
Furthermore, for all $v_j \in V'$ and $i=0,1$, we have the bound
$x^{j,i} \le w(G'[v_j=i]) \le w(G[v_j=i]) \le \frac{1}{4k} w(H \setminus G)$
because of the definition of $V'$ and $\hat{V}^i$.
By Corollary~\ref{coromulticolordisc}, we find a coloring $\chi\colon
V'\to\{0,1\}$ such that for each $1 \le i \le k$ it holds that
\[
\left|
\frac{1}{2}\sum_{v_j\in V'}\sum_{r=0}^1 x^{j,r}_i
-
\sum_{v_j \in V'} x^{j,\chi(v_j)}_i
\right| 
~~\le ~~
2k\max_{j,r} |x^{j,r}_i|
~~\le ~~
2k \frac{1}{4k} w_i(H \setminus G)
~~= ~~
\frac{1}{2} w_i(H \setminus G)
\]
and hence
\[
\sum_{v_j \in V'} x^{j,\chi(v_j)} \ge \frac{1}{2}\sum_{v_j \in
V'}(x^{j,0}+x^{j,1}) - \frac{1}{2} w(H \setminus G)
=\frac{1}{2}(w(G') - w(H \setminus G)).
\]
For $I$ being the truth assignment generated in this iteration it holds that
\begin{align}\label{eqn227736}
w(\{C \in G' \mid I(C) = 1\}) \ge
\sum_{v_j \in V'} x^{j,\chi(v_j)} \ge \frac{1}{2}(w(G') -w(H \setminus G)).
\end{align}
Furthermore, since $I$ and $I_o$ coincide on $V \setminus V'$, we have
\begin{align}
w(\{C \in H \setminus G' \mid I(C) = 1\}) &=
w(\{C \in H \setminus G' \mid I_o(C) = 1\})\label{eqn2343}\\
&\ge w(\{C \in H \setminus G \mid I_o(C) = 1\})\notag\\
& = w(\{H \setminus G\}).\label{eqn33438}
\end{align}
Thus we finally obtain
\begin{align*}
w(I) &\;=\; w(\{C \in H \setminus G' \mid I(C) = 1\}) + w(\{C \in G'\mid I(C) = 1\})\\
&\stackrel{\eqref{eqn227736}}{\ge} w(\{C \in H \setminus G' \mid I(C) = 1\}) +
\tfrac{1}{2}(w(G') - w(H \setminus G))\\
&\stackrel{\eqref{eqn2343}}{=} w(\{C \in H \setminus G' \mid I_o(C) = 1\}) +
\tfrac{1}{2}(w(G') - w(H \setminus G))\\
&\stackrel{\eqref{eqn33438}}{\ge} \tfrac{1}{2} w(\{C \in H \setminus G' \mid I_o(C) = 1\}) +
\tfrac{1}{2}w(G')\\
&\ge\; \tfrac{1}{2}w(I_o).\qedhere
\end{align*}
\end{proof}

%%%%%%%%%%%%%%%%%%%%%%%%%%%%%%%%%%%%%%%%%%%%%%%%%%%%%%%%%%%%%%%%%%%%%%%%%%%%%%%%%%%%%%%
%%%%%%%%%%%%%%%%%%%%%%%%%%%%%%%%%%%%%%%%%%%%%%%%%%%%%%%%%%%%%%%%%%%%%%%%%%%%%%%%%%%%%%%
%%%%%%%%%%%%%%%%%%%%%%%%%%%%%%%%%%%%%%%%%%%%%%%%%%%%%%%%%%%%%%%%%%%%%%%%%%%%%%%%%%%%%%%

\bibliographystyle{alpha}
\bibliography{tspbib}

\begin{thebibliography}{PMM09}

\bibitem[AL97]{al97}
S.~Arora and C.~Lund.
\newblock Hardness of approximations.
\newblock In D.~Hochbaum, editor, {\em Approximation Algorithms for {NP}-hard
  Problems}. PWS Publishing Company, Boston, 1997.

\bibitem[AW02]{AW02}
T.~Asano and D.~P. Williamson.
\newblock Improved approximation algorithms for {MAX} {SAT}.
\newblock {\em Journal of Algorithms}, 42(1):173--202, 2002.

\bibitem[BF81]{BF81}
J.~Beck and T.~Fiala.
\newblock {"Integer-Making" Theorems}.
\newblock {\em Discrete Applied Mathematics}, 3(1):1--8, 1981.

\bibitem[BM05]{BM05}
M.~Bl{\"a}ser and B.~Manthey.
\newblock Approximating maximum weight cycle covers in directed graphs with
  weights zero and one.
\newblock {\em Algorithmica}, 42(2):121--139, 2005.

\bibitem[CK99]{crka99}
P.~Crescenzi and V.~Kann.
\newblock A compendium of {NP} optimization problems.
\newblock URL: http://www.nada.kth.se/$\sim$viggo/problemlist/compendium.html,
  1999.

\bibitem[DS03]{DS03}
B.~Doerr and A.~Srivastav.
\newblock Multicolour discrepancies.
\newblock {\em Combinatorics, Probability \& Computing}, 12(4):365--399, 2003.

\bibitem[FNW79]{FNW79}
M.~L. Fisher, G.~L. Nemhauser, and L.~A. Wolsey.
\newblock An analysis of approximations for finding a maximum weight
  {H}amiltonian circuit.
\newblock {\em Operations Research}, 27(4):799--809, 1979.

\bibitem[GJ79]{GJ79}
M.~R. Garey and D.~S. Johnson.
\newblock {\em Computers and Intractability: A Guide to the Theory of
  NP-Completeness}.
\newblock W. H. Freeman \& Co., New York, NY, USA, 1979.

\bibitem[GLS88]{GLS88}
M.~Gr{\"o}tschel, L.~Lov{\'a}sz, and A.~Schrijver.
\newblock {\em {Geometric Algorithms and Combinatorial Optimization}}, volume~2
  of {\em Algorithms and Combinatorics}.
\newblock Springer, 1988.

\bibitem[Joh74]{Joh74}
D.~S. Johnson.
\newblock Approximation algorithms for combinatorial problems.
\newblock {\em Journal of Computer System Sciences}, 9(3):256--278, 1974.

\bibitem[KLSS05]{KLS+05}
H.~Kaplan, M.~Lewenstein, N.~Shafrir, and M.~Sviridenko.
\newblock Approximation algorithms for asymmetric {TSP} by decomposing directed
  regular multigraphs.
\newblock {\em Journal of the ACM}, 52(4):602--626, 2005.

\bibitem[Lau83]{lau83}
C.~Lautemann.
\newblock {BPP} and the polynomial hierarchy.
\newblock {\em Information Processing Letters}, 17:215--217, 1983.

\bibitem[Man08]{Man08}
B.~Manthey.
\newblock On approximating restricted cycle covers.
\newblock {\em SIAM J. Comput.}, 38(1):181--206, 2008.

\bibitem[Man09]{man09}
B.~Manthey.
\newblock On approximating multi-criteria {TSP}.
\newblock In S.~{Albers} and J.-Y. {Marion}, editors, {\em 26th International
  Symposium on Theoretical Aspects of Computer Science, STACS 2009}, pages
  637--648. Dagstuhl Research Online Publication Server, 2009.

\bibitem[Man11]{man11}
B.~Manthey.
\newblock Deterministic algorithms for multi-criteria {TSP}.
\newblock In {\em Proceedings of the International Conference on Theory and
  Applications of Models of Computation}, volume 6648 of {\em Lecture Notes in
  Computer Science}. Springer Verlag, 2011.
\newblock To appear.

\bibitem[MR09]{MR09}
B.~Manthey and L.~S. Ram.
\newblock Approximation algorithms for multi-criteria traveling salesman
  problems.
\newblock {\em Algorithmica}, 53(1):69--88, 2009.

\bibitem[PMM09]{PMM09}
K.~Paluch, M.~Mucha, and A.~Madry.
\newblock A 7/9 - approximation algorithm for the maximum traveling salesman
  problem.
\newblock In I.~Dinur, K.~Jansen, J.~Naor, and J.~Rolim, editors, {\em
  Proceedings of APPROX/RANDOM}, volume 5687 of {\em Lecture Notes in Computer
  Science}, pages 298--311. Springer Berlin / Heidelberg, 2009.

\bibitem[SBLL09]{SBLL09}
R.~Santana, C.~Bielza, J.~A. Lozano, and P.~Larra{\~{n}}aga.
\newblock Mining probabilistic models learned by {EDA}s in the optimization of
  multi-objective problems.
\newblock In {\em GECCO '09: Proceedings of the 11th Annual Conference on
  Genetic and Evolutionary Computation}, pages 445--452, New York, NY, USA,
  2009. ACM.

\bibitem[Sip83]{sip83}
M.~Sipser.
\newblock A complexity theoretic approach to randomness.
\newblock In {\em Proceedings of the 15th Symposium on Theory of Computing},
  pages 330--335, 1983.

\end{thebibliography}

%%%%%%%%%%%%%%%%%%%%%%%%%%%%%%%%%%%%%%%%%%%%%%%%%%%%%%%%%%%%%%%%%%%%%%%%%%%%%%%%
%%%%%%%%%%%%%%%%%%%%%%%%%%%%%%%%%%%%%%%%%%%%%%%%%%%%%%%%%%%%%%%%%%%%%%%%%%%%%%%%

\end{document}